\documentclass[sigconf,screen]{acmart}

\AtBeginDocument{%
  }

\acmYear{2024}\copyrightyear{2024}
\setcopyright{rightsretained}
\acmConference[ACM FAccT '24]{ACM Conference on Fairness, Accountability, and Transparency}{June 3--6, 2024}{Rio de Janeiro, Brazil}
\acmBooktitle{ACM Conference on Fairness, Accountability, and Transparency (ACM FAccT '24), June 3--6, 2024, Rio de Janeiro, Brazil}
\acmDOI{10.1145/3630106.3658952}
\acmISBN{979-8-4007-0450-5/24/06}

\acmConference[FAccT '24]{Make sure to enter the correct
  conference title from your rights confirmation emai}{June 3-6, 2024}{Rio de Janeiro, Brazil}

\begin{document}

\title{Structural Interventions and the Dynamics of Inequality}

\author{Aurora S. Zhang}
\orcid{0000-0002-6604-4765}
\affiliation{%
  \institution{Massachusetts Institute of Technology}
  \streetaddress{77 Massachusetts Ave}
  \city{Cambridge}
  \state{Massachusetts}
  \country{USA}
  \postcode{02139}
}
\email{aszhang@mit.edu}

\author{Anette E. Hosoi}
\orcid{0000-0003-4940-7496}
\affiliation{%
  \institution{Massachusetts Institute of Technology}
  \streetaddress{77 Massachusetts Ave}
  \city{Cambridge}
  \state{Massachusetts}
  \country{USA}
  \postcode{02139}
}
\email{peko@mit.edu}

\renewcommand{\shortauthors}{Zhang et al.}

\begin{abstract}
  Recent conversations in the algorithmic fairness literature have raised several concerns with standard conceptions of fairness. First, constraining predictive algorithms to satisfy fairness benchmarks may sometimes lead to non-optimal outcomes for disadvantaged groups. Second, technical interventions are often ineffective by themselves, especially when divorced from an understanding of structural processes that generate social inequality. Inspired by both these critiques, we construct a common decision-making model, using mortgage loans as a running example. We show that under some conditions, any choice of decision threshold will inevitably perpetuate existing disparities in financial stability unless one deviates from the Pareto optimal policy. This confirms the intuition that technical interventions, such as fairness constraints, often do not sufficiently address persistent underlying inequities. Then, we model the effects of three different types of interventions: (1) policy changes in the algorithm’s decision threshold, and external changes to parameters that govern the downstream effects of late payment for (2) the whole population or (3) disadvantaged subgroups. We show how different interventions are recommended depending on the difficulty of enacting structural change upon external parameters and depending on the policymaker’s preferences for equity or efficiency. Counterintuitively, we demonstrate that preferences for efficiency over equity may sometimes lead to recommendations for interventions that target the under-resourced group alone. Finally, we simulate the effects of interventions on a dataset that combines HMDA and Fannie Mae loan data. This research highlights the ways that structural inequality can be perpetuated by seemingly unbiased decision mechanisms, and it shows that in many situations, technical solutions must be paired with external, context-aware interventions to enact social change.
\end{abstract}



\received{22 Jan 2024}

\maketitle

\section{Introduction}

Algorithms are increasingly being utilized in socially impactful settings to guide decision-making. Several methods have been proposed for debiasing algorithms to ensure that they are non-discriminatory. Many of these methods consist of technical guidelines for algorithm development, which include altering parameters within machine learning models and constraining algorithms to satisfy certain fairness metrics \cite{barocasFairnessMachineLearning2023, mehrabiSurveyBiasFairness2021}.

This framing of algorithmic fairness has been criticized on several different fronts. Some have demonstrated that constraining algorithms to satisfy fairness criteria may lead to suboptimal welfare outcomes for members of both advantaged and disadvantaged subgroups. Others have argued that the focus on formal, technical interventions to ensure fairness distracts from more pressing concerns about how the predictive algorithm’s results are used to make unjust decisions that lead to undesirable social consequences. Still others are concerned with ways in which it is sometimes impossible to separate the effects of historical discrimination from existing social facts encoded in the training data; such historical inequities are often “baked in” to datasets. 

When designing policies to make algorithmic decision-making systems more fair and just, these concerns point us to ask several questions: first, what are the welfare effects and outcomes of these decision-making processes? Second, how can our understanding of the social contexts of these algorithms enable us to target the most crucial leverage points when it comes to policymaking? Third, to what extent are these algorithmic decision-making systems reifying unjust inequities that are a result of historical discrimination, and how can we design policies that ameliorate the effects of this injustice?

We borrow from feminist and political philosophy, as well as previous discussions in the algorithmic fairness literature \cite{kasirzadehAlgorithmicFairnessStructural2022, greenEscapingImpossibilityFairness2022}, to propose how we can use a framework of structural injustice and structural interventions to address these concerns in a formal model. We describe a class of decision structures, which we more fully characterize in section 3, in which applicants in a population undergo a decision process that is a function of their initial financial status. Depending on the result of this decision, applicants then gain or lose financial stability proportionally to their initial status. This is a type of feedback loop, in which the decision structure exacerbates existing inequalities present in the population over time \cite{readerModelsUnderstandingQuantifying2022}. We will use mortgage loans as a running example.

The main contributions of this paper are as follows:

\begin{itemize}

\item We create a quantitative model of the effects of structural interventions on a system with persistent inequality, based on theories of structural injustice previously qualitatively described in the algorithmic fairness literature.
\item We describe conditions under which any intervention upon decision thresholds will either perpetuate systemic inequality, or will harm the better-off group without any improvement to the plight of the worse-off group.
\item We characterize conditions under which structural interventions may be preferable to algorithmic ones, depending on a policymaker's preference for equity or efficiency.
\item We apply our framework and propose examples of structural interventions in the housing domain.

\end{itemize}

In Section 2, we lay out some existing debates within the algorithmic fairness community surrounding the inadequacy of predominant fairness models and situate our contributions within the context of related work. Then in Section 3, we set up our model and prove three propositions about its behavior. Section 4 describes and simulates the effects of different kinds of policy interventions upon sample distributions. Section 5 contains an empirical demonstration and simulation of these results using HMDA and Fannie Mae data, and we conclude in Section 6. 

\section{Background}

Structural injustice has been the focus of recent discussions in the algorithmic fairness literature \cite{kasirzadehAlgorithmicFairnessStructural2022, linArtificialIntelligenceStructurally2022, himmelreichAIStructuralInjustice2022, madaioFairnessStructuralJustice2022}. To illustrate an example of structural injustice, we begin with an enduring puzzle: as of 2024, discrimination based on identity categories like race, gender, and nationality is prohibited across domains like education, employment, lending, and housing. However, unequal outcomes, such as the wealth gap and the homeownership gap between Black and White Americans, or the income gap between men and women, still remain, and in many cases, have grown \cite{aliprantisWhatPersistenceRacial2019, shapiroRootsWideningRacial2013, darityWhatWeGet2018, kermaniRacialDisparitiesHousing2021}.

Why might this be the case? One explanation may be that membership in some identity categories may affect an individual’s preferences and behavior. A second explanation is implicit bias: decision-makers might still be affected by subconscious discrimination. Another explanation is that while present-day decision-makers may be unbiased, very little has been done to rectify past harm, the effects of which still pervade the present. A fourth explanation is that unequal outcomes persist, not only because of the biased actions of individual actors, nor because of injustice that exclusively happened in the past, but because of larger systems that are structurally biased against marginalized groups. 

Our definition of structural injustice focuses on the third and fourth explanations for how decision-making systems generate social inequality. First, decision-making systems are often trained on data that reflect unjust historical patterns, which causes them to reproduce status quo disparities that are a result of historical inequity \cite{maysonBiasBiasOut2019, greenFalsePromiseRisk2020}.
This problem cannot always be resolved by better data collection or by clarifying gaps between theoretical constructs and measurement \cite{friedlerImPossibilityFairness2021, jacobsMeasurementFairness2021}. In these most difficult cases, membership in the minoritized group can be constitutively related to features in the construct space that are, in turn, casually influential upon target outcomes \cite{huWhatRaceAlgorithmic2023, heringtonMeasuringFairnessUnfair2020a}.  Parallels can be drawn between this type of injustice and Young’s account of structural injustice, which often arises in situations where no actor is liable for blame or wrongdoing \cite{youngResponsibilityJustice2010}. Similarly, algorithmic decision-making systems can reify status quo inequalities without being biased or discriminatory in a traditional sense, and without being the result of blameworthy behavior by a discrete bad actor \cite{hoffmanWhereFairnessFails2019}.

Second, decision-making systems perpetuate inequality because they are situated in present-day, unjust social structures. A social structure can be thought of as a particular configuration of a network of causal influence between social phenomena \cite{haslangerWhatSocialStructural2016, sankaranStructuralInjusticeAnalytical2021}. Haslanger also describes a socio-structural explanation as one that answers not the what but the why of causal relationships. These investigations can be useful for understanding the effects of unjust social structures. For instance, a recidivism risk algorithm may inform a judge’s decision for pre-trial detainment. In many cases, we might ask not just “was this algorithm fair to sentence Black defendants in these cases?” but also, “what are the social structures that have created a causal connection between recidivism risk and imprisonment, or between imprisonment rate and community well-being?” These questions illuminate how the prison-industrial system structures networks of causal influence between algorithms and their social consequences.

While there is debate about moral responsibility for structural injustice \cite{mckeownStructuralInjustice2021}, we take as a starting point that data scientists and decision-makers are responsible for the social consequences of algorithmic decision-making systems \cite{kasirzadehAlgorithmicFairnessStructural2022, greenDataSciencePolitical2021}. We investigate the conditions under which decision-making systems perpetuate injustice and the mechanisms by which they do. 

Some solutions have included developing more careful guidelines for how algorithmic predictions affect decision-making, using distributive justice as a standard for making better decisions \cite{kupplerFairPredictionsJust2022a}. We draw in particular from Green’s concept of a “structural intervention” toward algorithmic justice \cite{greenEscapingImpossibilityFairness2022}, which alter the way that algorithmic predictions affect downstream social outcomes. In the recidivism case, this may mean decreasing incarceration time overall while increasing social services, so that an algorithmic prediction may not have such high-stakes social consequences. We model these structural responses by changing the parameters that govern the magnitude of downstream effects. \autoref{fig:structural} demonstrates this graphically: instead of focusing on debiasing a predictive algorithm at the site of the algorithm, we must examine the larger social structure that the algorithm is a part of, and if necessary, consider interventions upstream and downstream of the decision. This structural injustice framework describes how the entire causal chain may generate injustice without blameworthiness or liability at any particular point. 

\begin{figure}
    \centering
    \includegraphics[scale=0.6]{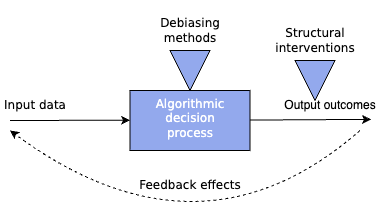}
    \caption{Structural interventions in the algorithmic decision-making pipeline change the effects of decisions on downstream consequences.}
    \label{fig:structural}
\end{figure}

\subsection{Fairness and the consequentialist and structural critique }

The socio-structural explanation for persistent inequality responds to many existing concerns raised by the algorithmic fairness literature, and in many cases, offers a solution to these existing dilemmas. 

First, the sociotechnical critique argues that many fairness interventions that focus exclusively on technical interventions at the site of the algorithm \cite{greenAlgorithmicRealismExpanding2020, greeneBetterNicerClearer2019} can obfuscate larger patterns of unjust social systems that these algorithms are a part of \cite{munnUselessnessAIEthics2023}. Instead, algorithms should be evaluated with their social context, with a particular focus on the interactions between the social and technical aspects of a decision-making system. This broader evaluative framework may investigate legal and political contexts, data collection processes, and social impacts of the decision-making systems \cite{selbstFairnessAbstractionSociotechnical2019a, cooperEmergentUnfairnessAlgorithmic2021b, dolataSociotechnicalViewAlgorithmic2021}.  In practice, some have suggested frameworks that use value chains or bottleneck theory \cite{jainAlgorithmicPluralismStructural2023, attard-frostEthicsAIValue2023a} that identify vulnerable or high-impact locations within the algorithmic pipeline that are most in need of intervention. 

Second, the consequentialist critique describes how seemingly fair procedural mechanisms may generate non-optimal outcomes for disadvantaged groups \cite{huFairClassificationSocial2020, corbett-daviesMeasureMismeasureFairness2018, sunBackfireEffectsFairness2022}, especially over a longer-term time horizon \cite{liuDelayedImpactFair2018a,damourFairnessNotStatic2020, puranikDynamicDecisionMakingFramework2022}. Additionally, compliance with legal standards of non-discrimination by “blinding” an algorithm to an individual’s protected attribute may also generate non-optimal outcomes \cite{mullainathanAlgorithmicFairnessSocial2018,kleinbergAlgorithmicFairness2018, kimRaceAwareAlgorithmsFairness2022}. Overall, the consequentialist perspective encourages evaluating outcomes rather than procedures \cite{cardConsequentialismFairness2020, kasyAlgorithmicBiasRacial2023, gillisInputFallacy2021a}, arguing that practitioners concerned with justice ought to directly study welfare for disadvantaged groups. Consequentialist-flavored case studies on algorithmic justice have been explored in domains such as education \cite{corbett-daviesMeasureMismeasureFairness2018}, lending \cite{chohlas-woodLearningBeFair2023}, and criminal justice \cite{huqRacialEquityAlgorithmic2019, corbett-daviesAlgorithmicDecisionMaking2017}. 

Third, impossibility theorems have proven that when there exist differences in the base rate of outcomes, decision-makers often cannot simultaneously satisfy multiple intuitive fairness metrics. For example, equal opportunity and calibration can be incompatible \cite{kleinbergInherentTradeOffsAlgorithmic2018}, as are equal opportunity and predictive parity \cite{chouldechovaFairPredictionDisparate2016}. Green discusses how this result is a reflection of the social world and its existing inequalities, demonstrating that “in an unequal society, decisions based in formal equality are guaranteed to produce substantive inequality.” \cite{greenEscapingImpossibilityFairness2022}. 

Fourth, the historical bias critique argues that, because of existing biases embedded in data, standard intuitions about what constitutes a fair decision frequently fail. Existing inequities can exist in the relevant construct space—for instance, racial discrimination can generate a predictive relationship between race and socioeconomic status that can have a real bearing on some outcome of interest \cite{friedlerImPossibilityFairness2021}. When these relationships exist due to historical discrimination, standard assumptions that are used to reason about intuitive fairness standards may fail \cite{fazelpourAlgorithmicFairnessNonideal2020a, heringtonMeasuringFairnessUnfair2020a, zimmermannProceedCaution2022}. Furthermore, some algorithmic decision-making systems may dynamically interact with a social system at multiple points, where input parameters at time $t$ may be dependent upon outputs of the algorithm itself at time $t-1$ or earlier. In these cases, algorithms may not simply perpetuate historical injustice but sometimes may even exacerbate discrepancies \cite{ensignRunawayFeedbackLoops2018,raabUnintendedSelectionPersistent2021,dobbeBroaderViewBias2018, fusterPredictablyUnequalEffects2022, linRichGetRicher2023}.

The structural explanation addresses concerns raised by these dilemmas. Our model responds to the sociotechnical critique, making explicit the values that policymakers might prefer \cite{cooperEmergentUnfairnessAlgorithmic2021b} and thinking beyond technical solutions for social change \cite{selbstFairnessAbstractionSociotechnical2019a}. The structural change model also addresses the consequentialist critique, explicitly defining a welfare function and assessing how policy changes may alter outcomes for both advantaged and disadvantaged populations. Next, the structural change model is based on Green’s suggestion that the concern with misaligned fairness metrics is largely a result of existing inequalities in the social world. Our model aims to capture situations in which such “base rate” differences can or cannot be eliminated. Finally, the dynamic nature of our model allows for the assessment of long-term consequences in a feedback-based setting, showing how existing inequalities might be perpetuated or mitigated over time. 

\subsection{Related Work}

Our work builds on a body of literature studying the dynamics of qualification rates. Liu et al. model a one-step lending process where individual’s score is affected by a selection mechanism. This model shows that constraining selection policies to satisfy certain fairness criteria, such as demographic parity or equal opportunity, can sometimes lead to greater declines in qualification status compared to the unconstrained policy \cite{liuDelayedImpactFair2018a}. Jorgensen et al. extend this result to several other fairness criteria \cite{jorgensenNotFairImpact2023}. Williams and Kolter model a similar loan approval setting as Liu et al., with slight modifications to the update function \cite{williamsDynamicModelingEquilibria2019}. In contrast, they show that unconstrained policies may increase inequality while fairness-constrained policies can sometimes lead to convergent outcomes. Our lending setting is inspired by these models, but instead of evaluating their behavior when constrained by specific fairness rules, we show that there are certain conditions where no possible policy constraints can decrease inequality without harming the advantaged group.

Other related work studies how qualification rates dynamically evolve over a longer time period \cite{mouzannarFairDecisionMaking2019}, how the model might change when individuals can choose to invest resources in improving their qualification status over time \cite{zhangHowFairDecisions2020}, or how within-group disparities are exacerbated by these decision systems \cite{sunBackfireEffectsFairness2022}. Like Mouzannar et al., Zhang and Tu, and Sun et al., we demonstrate that there exist conditions under which no threshold policy will lead to equality in the long term. We additionally examine a certain type of policy intervention that alters the relationship between late payment and financial penalties, which we call a structural intervention. Most of the works mentioned above explicitly assume a fixed background structure, examining only interventions on approval rates. Zhang and Tu allude to a certain type of “transitional intervention”  that increases the equilibrium approval rate for both groups, but do not show how such interventions may increase the equilibrium equity between these groups. We also study the differences between group-blind and group-specific interventions.

The most similar work that formally models structural interventions may be \cite{cruzcortesLocalityTechnicalObjects2022a}. Cruz Cort\'es et al. models three separate interventions: one that enforces fairness constraints, one that changes the initial score distribution of the two populations, and one where the disadvantaged population receives a boost. Their results largely agree with ours, showing that the latter two structural interventions are more effective in many situations. While Cruz Cort\'es et al. simulates only two structural interventions; we show how structural interventions of different magnitudes may affect outcomes. We then give an example of a policymaker’s utility function that places different weights on short-term equality and efficiency, and show the desirability of different interventions in each case. 

\section{Model}

\subsection{Single Time Step Model}

We use a loan approval setting as a running example. Consider a loan applicant pool that consists of two subgroups: an advantaged group $A$ and a disadvantaged group $D$. At time $t$, every applicant has some likelihood of paying off their requested loan. Let $p \sim \text{Ber}(\pi^i_t)$ be a random variable that corresponds to whether an individual pays off their loan or not, and suppose that the parameter $\pi^i_t$ depends on the individual's subgroup identity $i \in \{A,D\}$ and the time $t$. 

We assume that there exist two different distributions of probability repayment parameter $\pi^i_t$ for the two subgroups. Now consider a threshold decision rule that allocates loans to applicants: we choose $\beta^i \in [0,1]$ such that applicants with repayment probability parameter $\pi^i_t \geq \beta^i$ are granted the loan, while applicants with repayment probability parameter $\pi^i_t < \beta^i$ are denied. In similar settings, it has been demonstrated that a threshold rule is optimal, and that many fairness constraints, such as equal opportunity, can be enforced in terms of different thresholds \cite{liuDelayedImpactFair2018a, zhangHowFairDecisions2020}.

We examine the evolution of the distributions of $\pi^i_t$ over time, where $p \sim \text{Ber}(\pi^i_t)$:
\begin{equation*}\label{equation: updateeq}
\pi^i_{t+1} = \begin{cases}
\max\{\min\{\pi^i_t + kp - ck(1-p),1\},0\}
& \pi^i_t \geq \beta^i \\ 
\pi^i_t & \pi^i_t < \beta^i \end{cases}
\end{equation*}

$k$ is a scaling parameter, and $c$ measures the severity of financial penalty for a late payment. We will examine interpretations of $c$ more closely in the last section. At each time step, part of the population is approved and part of the population is denied. Denied individuals maintain the same likelihood of repayment; in this step, we hold external facts constant, assessing the effects of the decision-making process alone. Approved individuals either accrue benefit $k$, if they make a timely payment at time $t$, or penalized $kc$, if their payment is late. This update equation captures the intuition that timely payment of loans helps build credit, and in the home mortgage case, that homeownership is an avenue for wealth-building. This also captures the intuition that when $c >1$, the penalty for late payment is larger than the benefit accrued from repayment at each time step. Finally, the minimum and maximum bounds ensure that $p^i_{t+1} \in [0,1]$ for all $t$. 
 
Let $\mu^i_t = E[\pi^i_t]$. In expectation, this is
\begin{eqnarray*}\label{eq:mean}
    &&\mu_{t+1} = (1-F_{\pi^i_t}(\beta)) E[\max\{\min\{\pi^i_t + kp - ck(1-p), 1\}, 0\} | \pi^i_t\\ 
&&\geq \beta] + F_{\pi^i_t}(\beta)E[\pi^i_t| \pi^i_t< \beta]
\end{eqnarray*}
$F_{\pi^i_t}(\beta)$ refers to the cumulative distribution of $\pi^i_t$ evaluated at $\beta$. 

We assume that the policymaker is blind to each individual's group identity; they must set a universal $\beta$, so $\beta^A = \beta^D$. This assumption corresponds to legal standards such as disparate treatment law in mortgage lending, which prohibits using race, sex, nationality, and other protected characteristics as a reason for setting different approval rules \cite{kumarEqualizingCreditOpportunity2022, hellmanMeasuringAlgorithmicFairness2020}. 

\begin{definition}
A random variable A dominates B over the interval $[\gamma, \delta]$ if for all $x \in [\gamma, \delta]$, $F_A(x) \leq F_B(x)$, where $F$ is the cumulative distribution function.
\end{definition}

\begin{proposition}\label{structineq}
    Assume that for a given threshold $\beta \in [0,1]$, $\pi^A_t$ dominates $\pi^D_t$ for the interval $[\beta, 1]$, and that $\mu^A_t \geq \mu^D_t$. Then if $P(0 < \pi^A_{t+1},\pi^D_{t+1} < 1) = 1$, then $\mu^A_{t+1} - \mu^A_{t}  \geq \mu^D_{t+1} - \mu^D_{t}$.
\end{proposition}

All proofs for this and subsequent propositions can be found in \autoref{ap: proofs}. We make the assumption because of the linearity of our function, and because it is bounded at its endpoints. This assumption is not difficult to attain in practice: $\mu^i_t =1$ if and only if everyone in the population has probability of repayment $=1$, and $\mu^i_t=0$ if and only if everyone in the population has probability of repayment $=0$. These two cases are both trivial, and all possible policies will maintain those equilibria. 

In the intermediate cases, this proposition is related to the speed of convergence of $\mu^A_t$ and $\mu^D_t$. Not only are inequalities between the subpopulations maintained, they are deepened by the process. 

\begin{corollary}\label{structineqcoro}
    Assume that $\pi^A_t$ dominates $\pi^D_t$ for the interval $[0,1]$. Then no matter what choice of $\beta \in [0,1]$, if $P(0 < \pi^A_{t+1},\pi^D_{t+1} < 1) = 1$, $\mu^A_{t+1} - \mu^A_{t}  \geq \mu^D_{t+1} - \mu^D_{t}$.
    \end{corollary}

The corollary states that under certain conditions, any inequalities will be deepened by this distribution mechanism. No choice of $\beta$ will be able to avoid this outcome; that is, no change in policy will be enough to avoid perpetuating structural inequality. We show later that this is plausible in many empirical demonstrations. 

Next, we show what happens when we choose the optimal outcome. 

\begin{proposition}
    Assume that $\pi_t^A$ dominates $\pi_t^D$ over the interval $[0, 1]$, and $c, k \geq 0$. There exists an optimal decision threshold $\hat{\beta}$ that simultaneously maximizes $E[\pi^A_{t+1}]$ and $E[\pi^D_{t+1}]$. Then $\mu^A_{t+1} \geq \mu^D_{t+1}$ if we set the threshold to be $\hat{\beta}$.
\end{proposition}

To ensure that  $\mu^A_{t+1} - \mu^A_{t} = \mu^D_{t+1} - \mu^D_{t}$, we must lower the advantaged population mean without a corresponding increase in the disadvantaged population mean. It has been shown in similar situations that many group fairness criteria, such as demographic parity or equal opportunity, can be implemented through a decision threshold rule that permits different thresholds for different groups \cite{liuDelayedImpactFair2018a}. In other words, these fairness criteria are often not Pareto optimal. 

\subsection{Long Term Behavior}

We examine the steady-state behavior of this system while making one key assumption: that there exists no external influence on the dynamics of the system. 

\begin{proposition}
    Assume that $c,k \in \mathbb{Q}$, and let $\beta \in [0,1]$ be some arbitrary threshold. Then $P(\beta < \lim_{t \to \infty} \pi^i_t < 1) = 0$ for all group identities $i \in \{A,D\}$. 
\end{proposition}

The long-term steady-state of the process results in a bifurcation in the population within each identity subgroup: part of the population attains repayment probability $1$, while another part of the population is relegated to being below $\beta$ for perpetuity. In reality, external interventions can often propel members of the sub-$\beta$ group to a higher status. However, we are merely examining the long-term effects of this process in isolation, showing that \textit{no matter} the choice of $\beta$, low-resource members of the population may not be able to be helped by this process. In \autoref{ap: proofs} we demonstrate how we can compute the probability of certain stationary distributions given initial parameters.

\section{Interventions}

In the previous section, we showed that under certain assumptions, no decision-making policy will eliminate initial status inequalities. Now we examine several different types of interventions. Recall in subsection~3.1 
 that the policymaker can choose to set the threshold $\beta$. We consider situations under which the policymaker can change $c$ as well. 

Consider the following aggregate utility function for a population of size $n$ that is divided according to some identity label $i \in \{A, D\}$. Let $\alpha \in [0,1]$. First, let $\hat{\beta} = \text{argmax}_\beta \lim_{t \to \infty} E[\pi^i_t]$. Then let $\mathbb{E}[\hat{\pi^i_t}]$ denote the maximum mean repayment rate prior to any intervention, that is, the expected value at $t$ after applying the optimal decision threshold $\beta = \hat{\beta}$ after every time step. Let $\mathbb{E}[\pi^i_t]$ denote the mean repayment rate at time $t$ under some intervention policy. Then the policymaker's utility from an intervention is as follows: 

\begin{equation*}
U_t = -\alpha|\mathbb{E}[\pi_t^A] - \mathbb{E}[\pi_t^D]| + (1-\alpha) [|(\mathbb{E}[\pi_t^A] - \mathbb{E}[\hat{\pi_t^A}]| + |(\mathbb{E}[\pi_t^D] - \mathbb{E}[\hat{\pi_t^D}]|]
\end{equation*}

The first term corresponds to outcome parity: to what extent does it matter that the two groups have different outcome means? The second term corresponds to aggregate utility change across both groups, the difference between the optimal unconstrained policy and the post-intervention policy. This is a version of a utilitarian equity-efficiency tradeoff, where a higher $\alpha$ represents a preference for outcome equity over outcome efficiency and vice versa. We examine the effects of three possible interventions. In this particular formulation, we assume that the populations are of equal sizes, but it can be modified for the case where the populations are of different sizes.

\begin{itemize}
\item Policy 1: "beta only intervention" 

These policies only modify $\beta$. Depending on $\alpha$, we choose $\beta^i$ to maximize $U_t$. As demonstrated before, since $\mathbb{E}[\hat{\pi_t^A}] \geq \mathbb{E}[\hat{\pi_t^D}]$ in general, this generally means choosing $\beta$ such that it decreases the expected outcome of the advantaged group to bring it closer to parity with the expected outcome of the disadvantaged group.

\item Policy 2: "group-blind intervention"

This intervention modifies $c$. Let $\hat{c}$ be the pre-intervention level of $c$. Let $r$ correspond to the magnitude of change that we can make to $\hat{c}$. When $r$ is large, this means that we have the resources to decrease $\hat{c}$ significantly. We choose the new $c = \hat{c}-r\hat{c}/2$. When $\hat{c} < c$, then $\hat{\pi^i_t} \leq \pi^i_t$. Then, we also choose the optimal $\beta_i$ to maximize $U_t$. 

\item Policy 3: "group-conscious intervention"

We now assume that $c$ can be different for different groups: $c_i$ depends on the subgroup. This intervention exclusively modifies $c_D$. Because we are directing resources specifically to modify $c_D$, assume that $c_A = \hat{c_A}$ and $c_D = \hat{c_D} - r\hat{c_D}$. Similarly to policy 2, we permit changes in $\beta_i$ to maximize $U_t$.

\end{itemize}

Policy type 1 is a technical intervention within the existing decision structure, which holds $k$ and $c$ constant. Policy types 2 and 3 are structural interventions that depend on the amount of resources $r$ to modify the structural parameter $c$. $r$ can also be thought of, more generally, as the responsiveness of the social world to policies that attempt to alter $c$. 

In the following simulations, we start with two populations, one advantaged and one disadvantaged. We choose a sample size of $1000$ and simulate the effects of different policies and interventions on samples of these populations. First, we determine the optimal $\beta_A$ such that $\lim_{t \to \infty} \pi^A_t$ is maximized, and likewise for $D$. Then, we compare these outcomes $\lim_{t \to \infty} U_t$, with the outcomes generated when we choose certain interventions. 

\begin{figure*}
    \centering
    \includegraphics[scale=0.2]{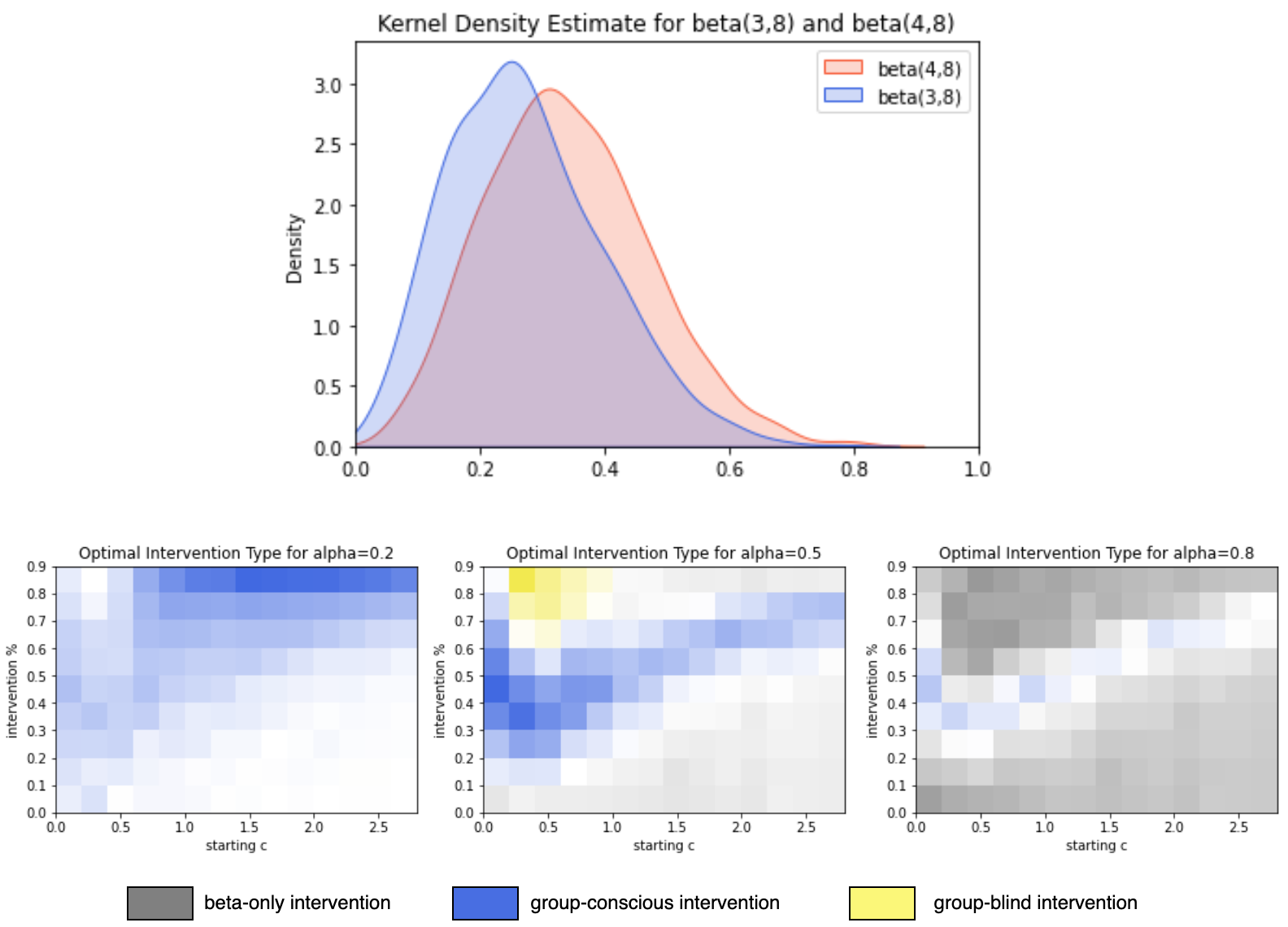}
    \caption{The top graph displays the distribution of $\pi^i_0$ for the advantaged (pink) population and the disadvantaged (blue) population. The three bottom graphs depict the recommended intervention for different levels of $\alpha$. The y-axis is $r$, as described in the earlier sections that enumerate each policy. We choose the policy that maximizes $\lim_{t \to \infty} U_t$. Higher opacity levels correspond to a stronger recommendation for a certain intervention over the other two (all effect sizes are scaled to 1). The lower left coordinate of each box is the anchor point for each box; for example, a yellow square with lower left coordinate (a,b) means that at starting $c=a$ and intervention \% $r = b$, the group-blind intervention is recommended.}
    \label{fig:single_graph}
\end{figure*}

In \autoref{fig:single_graph}, we sample $\pi^A_0$ from the $\text{beta}(4,8)$ distribution and $\pi^D_0$ from the $\text{beta}(3,8)$ distribution.\footnote{We chose the beta distribution because it is parametric with a $[0,1]$ range, making it well suited for modeling probabilities. There is precedent in the fairness literature for simulating risk with such distributions: see \cite{williamsDynamicModelingEquilibria2019, corbett-daviesMeasureMismeasureFairness2018, baumannEnforcingGroupFairness2022}.} We display a probability density estimate at time $t=0$ in the top graph. The bottom three graphs display the recommended policies after simulating policy outcomes over $t=20$ timesteps. They represent three different policymaker preferences: the leftmost graph shows the recommended interventions when the policymaker's $\alpha = 0.2$, that is, when they prefer outcome efficiency over equity, the middle graph with $\alpha=0.5$, which is a balanced preference, and the rightmost with $\alpha=0.8$ corresponds with a preference for equity over efficiency. For each of these individual graphs, the $x$-axis corresponds to the starting $c_i$ value, and the $y$-axis corresponds to the magnitude of change that is possible. In other words, a high intervention \% on the $y$-axis means that the starting $c_i$ can be significantly reduced by the policy intervention. 

\autoref{fig:single_graph} shows that group-conscious structural interventions are more strongly recommended when $\alpha=0.2$, especially when the intervention magnitude is large, while the beta-only interventions are sufficient when $\alpha=0.8$. These recommendations are highly contingent upon the starting distributions of $\pi^i_0$, however, and we see patterns emerge when comparing how recommendations differ depending on the initial distribution.

\autoref{fig:comparisons} depicts these results for 4 different distributions. Across all distributions, it appears that increasing $\alpha$ leads to a stronger recommendation for the $\beta$-only intervention. On the contrary, when $\alpha$ is low, the group-conscious interventions are more strongly recommended in general. $\alpha$ measures the degree one prefers "outcome equity" over "overall efficiency." Somewhat counterintuitively, this means that an increased preference for efficiency over outcome equity increases the strength of recommendations for group-specific interventions.

\begin{figure*}
    \centering
    \includegraphics[scale=0.3]{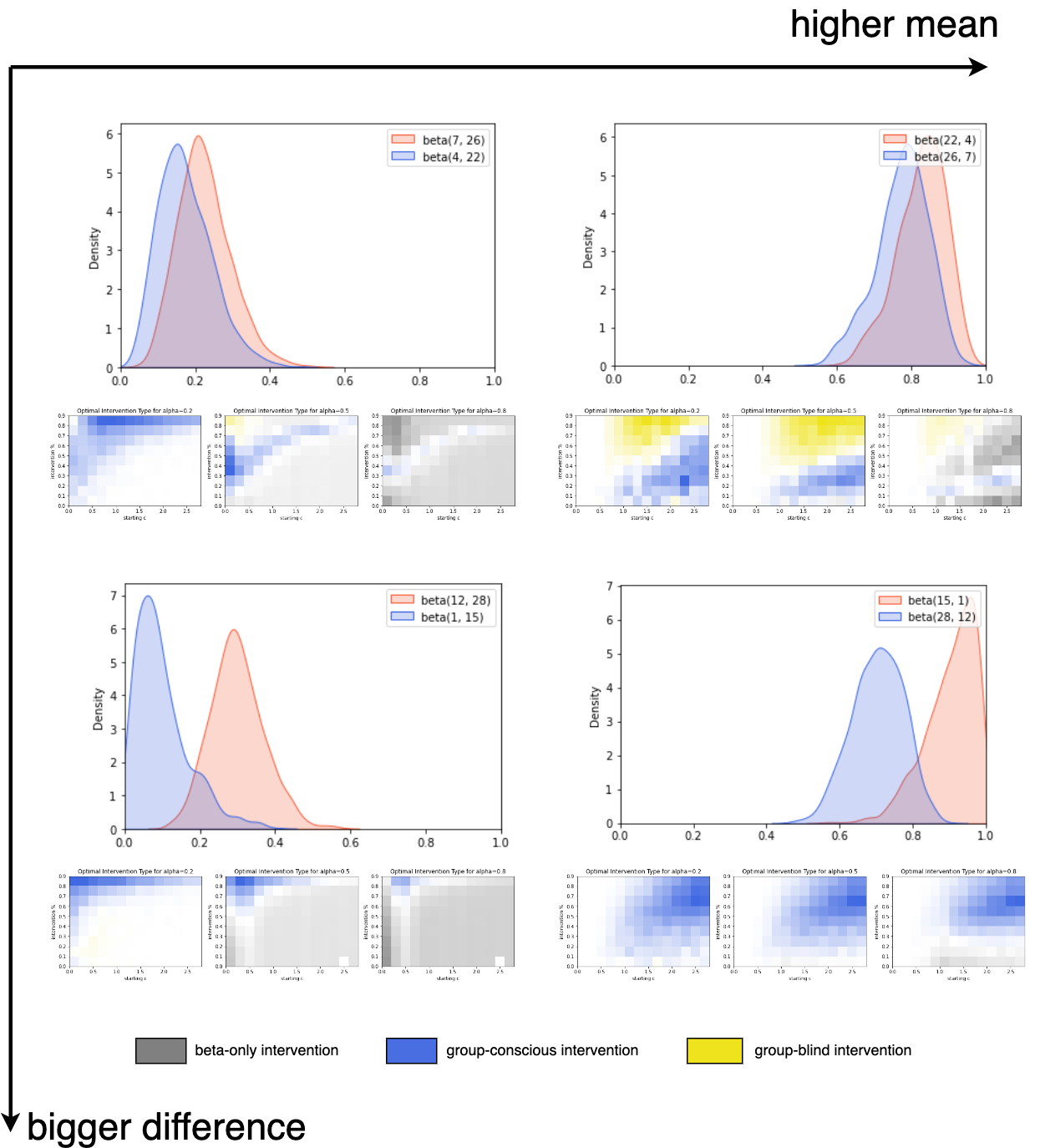}
    \caption{Each quadrant depicts starting distributions and recommended interventions at $\alpha=0.2,0.5,0.8$. Blue corresponds to the group-conscious intervention, yellow corresponds to the group-blind, and gray corresponds to beta-only. The lower left coordinate of each box is the anchor point for each box; for example, a blue square with lower left coordinate (a,b) means that at starting $c=a$ and intervention \% $r=b$, the group-conscious intervention is recommended. Opacity corresponds to the marginal utility of this intervention over the next best intervention type.}
    \label{fig:comparisons}
\end{figure*}

Other results depend more specifically on the starting distributions. First, we can look at the two graphs on the left side of \autoref{fig:comparisons} when the overall mean is low for both populations. Here there is a strong recommendation for group-conscious interventions when the starting $c$ level is low, or when there are enough resources such that the $c$ level can be reduced significantly for the disadvantaged population. This pattern shows up very clearly in the bottom left corner, when there is high intergroup difference. In these cases, only a large intervention will have any effect, because both populations are starting at a place where they are more likely to fail. The "band" pattern exhibited when the mean is low but the populations are close together can be explained by the fact that a too-large intervention would benefit the disadvantaged population to a point where they would be better off than the advantaged population. In this case, a preference for outcome parity would only lead to recommending the group-conscious intervention in the case where it leads exactly to parity and increased efficiency in outcomes. 

Next, when the overall mean is high, corresponding to the two graphs on the right side of \autoref{fig:comparisons}, a different pattern emerges. In the case when the populations have initially close starting positions, if the policymaker has fewer resources, and $c$ is high, the recommendation is to invest in group-conscious interventions. However, when $c$ is already low, or when the intervention is large, the recommended intervention is the group-blind intervention. This is likely because if the starting point for the disadvantaged population is fairly high, and it is easy for them to "catch up" to the advantaged population, a large intervention is not required for them to reach outcome parity and stability at the right-hand side of the distribution. In this case, it is most beneficial overall to improve the overall condition of both parties. 

But this pattern is different in the case where the population means are far apart, even when the overall mean is high. In this case, there is a strong preference for improving the plight of the disadvantaged group through structural means when $c_i$ can be significantly reduced. Since the advantaged population is already so well off, and the disadvantaged population can achieve much better outcomes through a c-type intervention, the recommendation is to invest all resources into the latter.

In general, the recommended intervention depends strongly on the policymaker's preference for outcome equity versus efficiency, and also depends on the starting conditions of the two populations. Note that each of these represents only the long-term effect of a one-time intervention. An expanded study may look at how dynamic interventions may respond to changing population distributions to improve overall conditions or to increase the speed of convergence of the two distributions in time.

\section{Empirical Demonstration}\label{section: empirical demo}

\subsection{Results}
We apply these methods to demonstrate how a mortgage loan approval procedure may perpetuate existing financial inequalities between Black and White applicants. We use two datasets. First is the publicly-available Home Mortgage Disclosure Act (HMDA) dataset from 2021,\footnote{https://ffiec.cfpb.gov/data-browser/} an individual loan-level dataset that records important demographic, geographic, and economic features of applicants, as well as labels of whether each application was approved or denied. It has a couple of limitations: although HMDA collects information on credit scores, this information is not released to the public. HMDA data also crucially does not contain data on default or late payment. 

The second dataset used is the publicly-available Fannie Mae Single Family Loan Performance Data\footnote{https://capitalmarkets.fanniemae.com/credit-risk-transfer/single-family-credit-risk-transfer/fannie-mae-single-family-loan-performance-data}. Fannie Mae (FNMA) and Freddie Mac are government-sponsored enterprises that buy conforming mortgages on the secondary mortgage market from lending institutions. As of 2020, Fannie Mae and Freddie Mac own approximately sixty percent of conforming mortgage loans in the U.S. \cite{WhatTypesMortgages2021} The Fannie Mae dataset contains a nationwide subset of all mortgage loans owned by Fannie Mae, originated between 2000 and 2022. FNMA contains loan-level information on the financial features of borrowers at loan origination. FNMA also contains time series, loan-level data on repayment on a monthly basis. For ease of data processing, we limited the analysis to mortgages originating in Massachusetts, and we used a subset of the FNMA data from the first two quarters of 2017. 

Our goal was to determine the distribution of risk of mortgage non-repayment across different subpopulations in the 2021 HMDA dataset. To do this, we found shared features between the HMDA and FNMA dataset: loan amount, loan-to-value ratio, debt-to-income ratio, and number of units in the property. Then, we created a logistic regression to predict risk of late payment, which we then trained on the FNMA dataset. Using this model, we predicted risk of late payment in the 2021 HMDA sample. Because there were relatively few features used, this model may not have extremely high predictive accuracy; however, it is useful for illustrative purposes to demonstrate how different distributions of risk across different populations may result in sustained inequality. For additional information on the datasets and methodology, see \autoref{ap: empirical}.

\begin{figure*}
    \centering
    \includegraphics[scale=0.4]{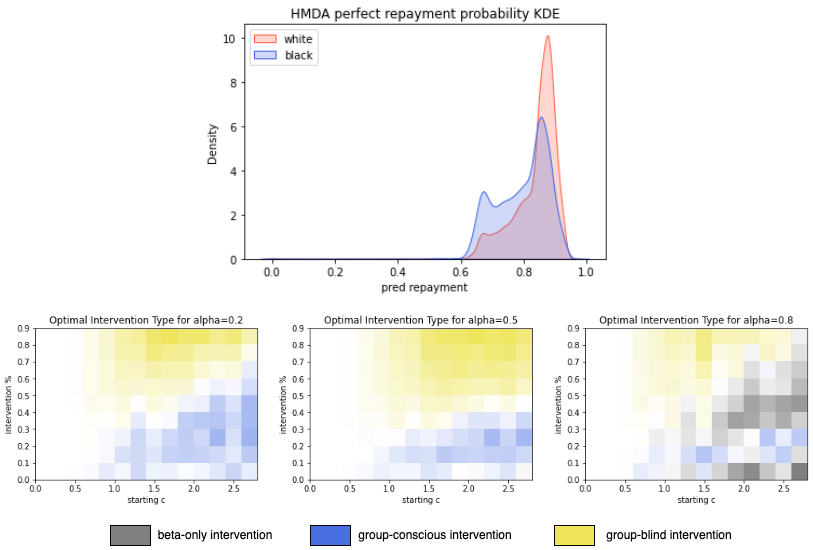}
    \caption{The top panel shows the empirical kernel density estimate of the distribution predicted risk of late payment in the HMDA Massachusetts 2021 data. The bottom three graphs show the recommended interventions for different $\alpha$, $c$, and intervention size.}
    \label{fig:hmda demo}
\end{figure*}

The predicted repayment values are displayed in the top panel of \autoref{fig:hmda demo}. This distribution fulfills the condition of stochastic dominance described above, and an empirical test across all $x \in [0,1]$ at a step size of $0.01$ verifies this assumption. See \autoref{fig:emp cdf} in \autoref{ap: empirical} for a visual demonstration. Since this is true, if we let $E[\pi^W_t]$ to be the expectation of repayment probability for the White population and $E[\pi^B_t]$ to be the expectation of repayment probability for the Black population, \autoref{structineqcoro} tells us that for all possible $\beta$, $E[\pi^W_t] \geq E[\pi^B_t]$. 

We then simulated the repayment status of each loan for $t=20$ timesteps to model the dynamics of population inequality over time. \autoref{fig:hmda demo} displays the recommended intervention for each value of $c$ and the policymaker's choice of $\alpha$. Notice that when $c<1$, the choice of intervention becomes relatively unimportant, as both distributions converge to $1$ without any external intervention. When $c>1$, we are split into two cases. When we have the resources to perform a large intervention that significantly reduces $c$, the optimal policy is a group-blind intervention that benefits both groups. However, when resources are limited and the intervention size is small, if we prefer outcome parity, the beta-only intervention is optimal. If we prefer efficiency, the group-conscious intervention that lowers $c$ for one group only is optimal. 

There are several caveats to the model: first, we calculated repayment risk based on a single late payment. While it is true that a single late payment can substantially lower credit score \cite{demyanykImpactMissedPayments2017}, penalties for late payments on mortgage loans tend to be relatively lenient, compared to auto or credit loans \cite{lanningWhatAreConsequences2020}. Mortgage default tends to result in more serious consequences, but there were not enough defaults in the dataset to construct a meaningful prediction model. 

Second, our model evaluates the entire population at every time step to determine whether they are above or below the given threshold, “accepting” or “rejecting” them each time. In practice, individuals do not have to periodically re-apply for a loan, and furthermore, rejected applicants are not likely to re-apply for a loan so soon after a rejection. Regardless, the outcomes of the model conform to the likely result, which is that approved households remain approved at every time step unless they fall below some threshold, after which they remain perpetually in the rejection state. Future work may model the effects of interventions on rejected members below the threshold.

Third, this model imperfectly captures the relationship between wealth accrual from homeownership and loan repayment risk. Households accrue wealth the longer they own their homes, but wealth may have heterogeneous or nonlinear effects on repayment risk \cite{sadhwaniDeepLearningMortgage2021}.  

Finally, notice that this sample only includes individuals who applied for a mortgage loan. If one’s concern is only with the fairness of the model, this is not a problem. However, for the purposes of understanding larger-scale patterns of housing injustice, our sample likely understates the racial gap in mortgage risk. This is because the mortgage applicant pool is subject to self-selection bias. Many individuals who want to own a home, who do not have enough financial assets to do so, will not apply.

\subsection{Discussion and Policy Recommendations}

There exists persistent racialized inequalities in housing despite the de jure elimination of segregation and discrimination in the U.S. \cite{squiresFightFairHousing2017}. Homeownership is one of the primary ways to build wealth in the U.S., and the gap between the percentage of white Americans who are homeowners and the percentage of Black Americans who are homeowners has increased since the 1960s \cite{herbertHomeownershipStillEffective2013,choiExplainingBlackWhiteHomeownership2019, kueblerClosingWealthGap2013}. This is a result of the legacy of historical redlining and segregation, as well as present-day structural barriers to homeownership for Black Americans. This can be seen in the 2008 financial crisis: well-intentioned policies that attempted to reduce barriers to homeownership have nevertheless harmed communities they intended to help by subjecting largely low-income, single-head-of-household Black women to predatory and exploitative loans \cite{taylorRaceProfitHow2019}.

We chose this application domain because it demonstrates how loan approval algorithms are embedded in an unjust social context. The statistical association between race and financial capital cannot be eliminated only by better data collection, blinding methods, or fairness constraints. Black Americans who, on average have less wealth and income, are less likely to be able to afford and maintain a mortgage \cite{deyRoleCreditAttributes2022, hilberExplainingBlackWhite2008, renFrameworkExplainingBlackWhite2020}. In line with recent work on reparative algorithms, we suggest a framework of how to evaluate policies that actively seek to remedy these injustices, rather than merely engaging in formalist procedures that gesture at fairness without any improvement of material circumstances \cite{soFairnessReparativeAlgorithms2022a, davisAlgorithmicReparation2021}.

What kind of policies will close the homeownership and wealth gap? On a pure egalitarian view, outcome parity can be achieved by harming the advantaged group enough such that their average outcome is closer to the disadvantaged average. This is the recommended policy when we place a heavy weight on bare egalitarianism; a simple alteration of the decision threshold may lead to worse outcomes for the advantaged group. However, this outcome may not be intuitively desirable. This intuition powers the well-known leveling-down objection to egalitarianism \cite{parfitEqualityPriority2002}. Prioritarianism is an attractive alternative to egalitarianism which argues that benefits to the worst-off members of society should be weighed more heavily than benefits to advantaged members \cite{temkinEgalitarianismDefended2003}. Prioritarian concerns apply most saliently when we are faced with the question of reallocating benefits from a more advantaged group to a less advantaged group. But in this case, the act of deliberately harming members of the advantaged group does not entail material benefit to the worse-off group. 

Pure egalitarianism, and its non-consequentialist counterpart, demographic parity, have drawn a plethora of criticism \cite{dworkFairnessAwareness2011, abu-elyounesContextualFairnessLegal2020}. If we wish to design algorithms that advance material justice and want to reject pure egalitarian approaches towards equality, we must alter the decision structure itself.

Our model includes one such approach. In practice, several policy mechanisms can lower $c$. For instance, late mortgage payments have a significant impact on credit score \cite{demyanykImpactMissedPayments2017}, which may impact one’s ability to access credit in other markets. Mitigating the drop of credit score upon a single late payment may lower the financial impact of late payment, as credit ratings can causally impact future creditworthiness \cite{mansoFeedbackEffectsCredit2013, purnanandamCanCreditRating2023}. Another intervention may involve making mortgage forbearance programs with favorable interest rates more widely available during times of financial crisis. Mortgage forbearance programs were used by largely minority and low-income borrowers during first couple years of the COVID-19 pandemic, and such programs reduced inequalities, helped maintain cash reserves for families affected by unemployment, and allowed borrowers to pay off other pressing debts \cite{anInequalityTimeCOVID192022, farrellDidMortgageForbearance2020, santiagoRoleForbearanceSustaining2023}.

There is evidence that $c$ for mortgage loans is already lower than it is for some other forms of debt \cite{lanningWhatAreConsequences2020}. Mortgage debts, unlike automobile debts or credit card debts, may not result in immediate foreclosure, whereas credit card loans or auto loans tend to default much more quickly. In the housing domain, our model may be more useful for modeling the housing process for renters: a single missed rent payment is much more likely to result in an eviction filing. Eviction may significantly negatively affect many aspects of a renter’s life, including health outcomes \cite{vasquez-veraThreatHomeEviction2017, benferEvictionHealthInequity2021}, employment prospects \cite{collinsonEvictionPovertyAmerican2024, desmondHousingEmploymentInsecurity2016}, and future housing outcomes \cite{tsaiLongitudinalStudyHousing2021}. A structural change in this area would decouple the effect of late payments and the cascade of harmful consequences that result from eviction.

Mortgage forbearance programs or eviction harm reduction programs are policy changs that we call “structural modifications to the decision system.” Another structural intervention may increase the affordable housing stock, an intervention which operates upstream of the decision process. Mathematically, we can model this as a shift in the distribution of risk across all populations—everybody is more likely to find affordable housing. Both types of interventions can be modeled by changing the parameters that govern the behavior of the dynamic system, and future work can evaluate the tradeoffs and the impacts of upstream and downstream interventions in the housing space.

However, such programs can be limited in their imaginative capacity. One structural intervention that is harder to model may involve the restructuring of debt markets entirely. For instance, Herzog calls into question the ethics of private debt markets, arguing that they are a mechanism for dynamically reinforcing injustice \cite{herzogEpistemicDivisionLabour2020}. Such sweeping social reconfigurations are currently difficult to evaluate, model, or implement, and more research is necessary to understand how to design and assess such alternative structures.

\section{Conclusion}

In this project, we have modeled a common decision-making mechanism in a loan approval setting and shown how certain initial inequitable conditions can be perpetuated over time. Then, inspired by accounts of structural injustice in algorithmic decision-making \cite{kasirzadehAlgorithmicFairnessStructural2022, greenEscapingImpossibilityFairness2022}, we imagine how a structural intervention might be incorporated into this formal model as alterations upon parameters that govern the relationship between an algorithmic decision and its social consequences. We then demonstrate the effects of policies that target these parameters, showing how different policies may be recommended depending on a policymaker's preference for equity or efficiency. In particular, we simulate the effects of resource limitations and constraints on the magnitude of possible social change on our recommended interventions. We find that in many cases, structural interventions are widely preferred and are often recommended over exclusively technical interventions, especially when efficiency is valued over bare outcome equality.

Our project imagines a type of structural intervention upon a penalty parameter, which represents only one way in which we can model the effects of structural change. More research is necessary to assess the types of policies which cause social change. Further research can also create more sophisticated models of larger-scale social intervention.

We wish to conclude by emphasizing that such interventions can hold the key to addressing many of the conundrums and contradictions that the fairness community has grappled with for the past few years: why is it that so many common, intuitive conceptions of fairness are incompatible with each other in practice? And why is it that adhering to these fairness metrics may not materially benefit disadvantaged populations in the long run? Industry standard fairness criteria should not be rejected entirely; they operate at a different scale, and serve a different function, from structural interventions that aim to advance justice in various domain areas of social importance. However, it is important to clarify the extent to which fairness guidelines can deliver on desirable outcomes. If policymakers are interested in more equitable, efficient social outcomes, structural interventions may be most suited towards advancing those ends.

\section{Ethics and Positionality Statement}

The publicly-available datasets from HMDA and Fannie Mae both contain fully anonymized entries. At the beginning of our study, we considered using a fuzzy matching procedure to link HMDA entries to the monthly default data available through Fannie Mae. However, because crucial features are missing from both datasets, it quickly became clear that both these datasets are designed such that matching would be impossible to do, out of privacy concerns. For instance, HMDA does not contain credit score as a feature, and Fannie Mae contains only the first three digits of a zip code and no demographic information, making it very difficult to perform any matching based on geographical characteristics, demographic characteristics, or financial characteristics. Any ethical concerns regarding the privacy of the individuals in these datasets are largely mitigated by the fact that these datasets do not contain granular enough information to perform any matching procedure. 

The authors of this paper work at a North American research university and the empirical application of this project is largely based on a U.S.-centric understanding of the racialized dynamics of the housing market. Additionally, the methods employed by the authors are largely quantitative and engineering-based, reflecting the authors’ training and background.

\begin{acks}
We would like to thank the anonymous FAccT reviewers for their feedback and suggestions for this paper. We thank members of the MIT Systemic Racism \& Computation working group in housing, particularly Wonyoung So and Catherine D'Ignazio, for feedback and discussions. Additionally, we would like to thank Sally Haslanger, William Li, Mikhail Yurochkin, Kristjan Greenewald, Felice Frankel, Anthony Bau, and Shomik Jain for helpful conversations that have informed the development of this paper. This research has been generously supported by the MIT-IBM Watson AI Lab. 
\end{acks}

\bibliographystyle{ACM-Reference-Format}
\bibliography{references}

\clearpage

\onecolumn
\appendix

\section{Proofs} \label{ap: proofs}

\begin{definition}
A random variable A dominates B over the interval $[\gamma, \delta]$ if for all $x \in [\gamma, \delta]$, $F_A(x) \leq F_B(x)$, where $F$ is the cumulative distribution function.
\end{definition}

\begin{proposition}\label{structineq}
    Assume that for a given threshold $\beta \in [0,1]$ such , $\pi^A_t$ dominates $\pi^D_t$ for the interval $[\beta, 1]$, and that $\mu^A_t \geq \mu^D_t$. Then if $P(0 < \pi^A_{t+1},\pi^D_{t+1} < 1) = 1$, then $\mu^A_{t+1} - \mu^A_{t}  \geq \mu^D_{t+1} - \mu^D_{t}$.
\end{proposition}

\begin{proof}

    For ease of notation, let $\pi^+_{t+1} = \max\{\min\{\pi^i_t + k p- ck(1-p), 1\}, 0\}$. First, we use the tower rule:
\begin{align*}
    E[\pi^+_{t+1}] &=
    E[ \pi^+_{t+1}| \pi^i_t + kp - ck(1-p)) \leq 0] p(\pi^i_t + k (\pi^i_t - c(1-\pi^i_t)) \leq 0) \\
    & + E[\pi^+_{t+1} | 0<\pi^i_t + kp - ck(1-p) < 1] p(0 < \pi^i_t - c(1-\pi^i_t)) < 1) \\
    & + E[\pi^+_{t+1} | 0<\pi^i_t + kp - ck(1-p) \geq 1]  p(\pi^i_t - c(1-\pi^i_t)) \geq 1) 
\end{align*}
which, by our assumption, simplifies to
\begin{align*}
    E[\pi^+_{t+1}]
    &= E[\pi^+_{t+1} | 0<\max\{\min\{\pi^i_t + kp - ck(1-p), 1\}, 0\} < 1] \\
    &= E[\pi^i_t + kp - ck(1-p)| 0<\pi^i_t + kp - ck(1-p) < 1]
\end{align*}

    Now we want to show that $\mu^A_{t+1} - \mu^A_{t}  - (\mu^D_{t+1} - \mu^D_{t}) \geq 0.$ Let the conditional expectation $\mu_t^+(\beta) = E[\pi^i_t | \pi^i_t \geq \beta]$ and $\mu_t^-(\beta) = E[\pi^i_t | \pi^i_t < \beta]$.
    \begin{align*}
        \mu_{t+1} &= (1-F_{\pi_t}(\beta)) E[\pi^i_t + kp - ck(1-p)| 0<\pi^i_t + kp - ck(1-p) < 1, \pi^i_t \geq \beta] + F_{\pi_t}(\beta)E[\pi_t | \pi_t < \beta] \\
        &= (1-F_{\pi_t}(\beta)) [\mu_t^+(\beta) + k (\mu_t^+(\beta)) - ck(1-\mu_t^+(\beta))] + F_{\pi_t}(\beta) \mu_t^-(\beta) \\
        &= [(1-F_{\pi_t}(\beta))(\mu_t^+(\beta)) + F_{\pi_t}(\beta) \mu_t^-(\beta)] + (1-F_{\pi_t}(\beta))k (\mu_t^+(\beta) - c(1-\mu_t^+(\beta))) \\
        &= \mu_t + (1-F_{\pi_t}(\beta))k (\mu_t^+(\beta) - c(1-\mu_t^+(\beta)))
    \end{align*}
    The second equality follows from the first, again by the tower rule. Substituting this in, we have
    \begin{equation*}
        \mu^A_{t+1} - \mu^A_{t}  - (\mu^D_{t+1} - \mu^D_{t}) = \\
        (1-F_{\pi^A_t}(\beta))k (\mu_t^{A+}(\beta) - c(1-\mu^{A+}_t(\beta))) + (1-F_{\pi^D_t}(\beta))k (\mu_t^{D+}(\beta) - c(1-\mu^{D+}(\beta))
    \end{equation*}
    Now based on our assumption, we know that $F_{\pi^A_t}(\beta) \leq F_{\pi^D_t}(\beta)$, so $(1-F_{\pi^A_t}(\beta)) \geq (1- F_{\pi^D_t}(\beta))$. The assumptions $\mu^A_t \geq \mu_t^D$ and the dominance assumption guarantee that the conditional mean $\mu^{A+}(\beta) \geq \mu^{D+}(\beta)$. This ensures that $ \mu^A_{t+1} - \mu^A_{t}  - (\mu^D_{t+1} - \mu^D_{t}) \geq 0$, as desired.

    \end{proof}

\begin{proposition}
    Assume that $\pi_t^A$ dominates $\pi_t^D$ over the interval $[0, 1]$, and $c, k \geq 0$. There exists a optimal decision threshold $\hat{\beta}$ that simultaneously maximizes $E[\pi^A_{t+1}]$ and $E[\pi^D_{t+1}]$. Then $\mu^A_{t+1} \geq \mu^D_{t+1}$ if we set the threshold to be $\hat{\beta}$.
\end{proposition}

\begin{proof}

First, we show that there exists a $\hat{\beta}$ such that it maximizes the value of
$$
\mu_{t+1} = (1-F_{\pi_t}(\beta)) E[\max\{\min\{\pi^i_t + kp - ck(1-p), 1\}, 0\} | \pi^i_t > \beta] + F_{\pi^i_t}(\beta)E[\pi^i_t | \pi^i_t < \beta]
$$

We know that $\pi^i_t$ is a continuous random variable, and the cdf $F_{\pi^i_t}$ is continuous. Furthermore, $E[\pi^i_t | \pi^i_t < \beta]$ is a conditional expectation of a continuous random variable, and thus is also continuous. The only term remaining that we have to assess is $E[\max\{\min\{\pi^i_t + kp - ck(1-p), 1\}, 0\}]$. We can rewrite this, knowing that $p \sim \text{Ber}(\pi^i_t)$: 
\begin{align*}
E[\max\{\min\{\pi^i_t + kp - &ck(1-p), 1\} | \pi^i_t > \beta] = \int_{\beta}^1 \max\{\min\{x + kp - ck(1-p), 1\}, 0\} f_{\pi_t | \pi_t > \beta}(x) dx \\
&= \int_{\beta}^1  [x\max\{\min\{x + k, 1\}, 0\} +
    (1-x)\max\{\min\{x - ck, 1\}, 0\}]f_{\pi_t | \pi_t > \beta}(x) dx \\
&= E[\pi^i_t\max\{\min\{\pi^i_t + k, 1\}, 0\}+(1-\pi^i_t)\max\{\min\{\pi^i_t - ck, 1\}, 0\} | \pi^i_t > \beta]
\end{align*}

This is the conditional expectation of a function of a continuous variable $\pi^i_t$, where the function consists of products, maximums, and minimums that are all continuous. Therefore we know $\mu_{t+1}$ is a continuous function of $\beta$. By the extreme value theorem, $\mu_{t+1}$ attains a maximum value. That is, there exists $\beta$ such that $\mu_{t+1}(\beta)$ attains a maximum.

    Now we show that the maximum does not depend on the distribution of $\pi_t$. We can write $\mu_{t+1}$ as a conditional sum:
    \begin{align*}
        \mu_{t+1}(\beta) &= \int_{\beta}^1 \max\{\min\{x + kp - ck(1-p), 1\}, 0\} f_{\pi_t}(x) dx + \int_{0}^\beta x f_{\pi_t }(x) dx \\
        &= \int_{\beta}^1 [x\max\{\min\{x + k, 1\}, 0\} +
    (1-x)\max\{\min\{x - ck, 1\}, 0\}]f_{\pi_t }(x) dx + \int_{0}^\beta xf_{\pi_t }(x) dx
    \end{align*}
    The function $g(x) = x\max\{\min\{x + k, 1\}, 0\} +
    (1-x)\max\{\min\{x - ck, 1\}, 0\}$ is nondecreasing in $x$ if $k,c \geq 0$ and linear between $0$ and $1$. So either $g(x) \leq x$ for all $x \in [0,1]$ or there exists some minimal $x_0 \in [0,1]$ such that for all $x > x_0$, $g(x) > x$.

    In the first case, set $\beta = 1$, and in the second, set $\beta = \min\{\max\{x_0,0\},1\}$. We can see that this would maximize the sum above. Finding $\beta$ here also does not depend on the distribution of $\pi_t$, so the optimal $\hat{\beta}$ is the same for both the advantaged and disadvantaged populations. 

    Second, we prove that $\mu^A_{t+1} \geq \mu^D_{t+1}$. Because of the dominance assumption, we know that 
    $$
    E[\pi^A_{t+1} | \pi^A_{t+1} < \hat{\beta}] -E[\pi^D_{t+1} | \pi^D_{t+1} < \hat{\beta}] = \int_{0}^{\hat{\beta}} x f_{\pi^A_t | \pi^A_t < \hat{\beta} }(x) dx - \int_{0}^{\hat{\beta}} x f_{\pi^D_t |\pi^D_t < \hat{\beta}}(x) dx \geq 0 
    $$
    Recall that $g(x)$ is nondecreasing in $x$. So
    \begin{align*}
        &E[\pi^A_{t+1} | \pi^A_{t+1} \geq \hat{\beta}] -E[\pi^D_{t+1} | \pi^D_{t+1} \geq \hat{\beta}] \\
        &=  \int_{x_1}^{x_2} [x\max\{\min\{x + k, 1\}, 0\}+(1-x)
    \max\{\min\{x - ck, 1\}, 0\}][f_{\pi^A_t}(x)-f_{\pi^D_t}(x)] dx  \geq 0
    \end{align*}

    Because of our dominance assumption, and because $[x\max\{\min\{x + k, 1\}, 0\}+(1-x)
    \max\{\min\{x - ck, 1\}, 0\}]$ is linear and nondecreasing in $x$,  this expression $\geq 0$.
    Therefore $E[\pi^A_{t+1} | \pi^A_{t+1} \geq \hat{\beta}] -E[\pi^D_{t+1} | \pi^D_{t+1} \geq \hat{\beta}]$. Since, as shown above, it is also true that $ E[\pi^A_{t+1} | \pi^A_{t+1} < \hat{\beta}] -E[\pi^D_{t+1} | \pi^D_{t+1} < \hat{\beta}]$, we can conclude over all that $E[\pi^A_{t+1}] \geq E[\pi^D_{t+1}]$.
     \end{proof}

\begin{lemma}
    Consider a random walk in one dimension bounded between $[x_0,x_1]$, starting at some point between $x_0$ and $x_1$. One moves some distance to the right, $y\geq 0$ with some probability $\pi$, and some distance to the left, $z\leq 0$ with probability $1-\pi$. Suppose that any point $p > x_1$ is all one absorbing state, and any point $p < x_0$ is another absorbing state, and suppose $y, z \in \mathbb{Q}$. Then there exists a finite amount of positions between $x_0$ and $x_1$, plus the two absorbing states, that the random walk could attain with probability $>0$.
\end{lemma}

\begin{proof}
    Since $y,z \in \mathbb{Q}$, we can represent them as $y=\frac{a}{b}$ and $z=\frac{c}{d}$, where $a,b,c,d \in \mathbb{Z}$. Then at any arbitrary time $t \geq 0$, we have taken $n$ steps to the right and $m$ steps to the right, where $m,n \in \mathbb{N}$. So the position we are at at time $t$ is either $x_0, x_1$, or can be represented as 
    $\frac{ma}{b} + \frac{nc}{d} = \frac{mad + ncb}{bd}$. There are a finite number of multiples of $\frac{1}{bd}$ between $x_0$ and $x_1$, so there are a finite number of positions that the random walk can attain.
\end{proof}

\begin{proposition}
    Assume that $c,k \in \mathbb{Q}$, and let $\beta \in [0,1]$ be some arbitrary threshold. Then $P(\beta < \lim_{t \to \infty} \pi_t < 1) = 0$. 
\end{proposition}

\begin{proof}
Consider a Markov chain, and let $\pi_0$ be the starting state. Since $c, k \in \mathbb{Q}$, this means there are finitely many states that $\pi_t$ can attain between $[0, 1]$, as per the above lemma. Suppose there are $n$ states, and $m$ of them are absorbing states that correspond to $0, 1$, and any point $< \beta$. Then we can describe this with a $n \times n$ matrix $S$, where each entry $s_{ij}$ corresponds to the transition probability between state $x_i$ and state $x_j$. This transition matrix can also be written as the block matrix
$$
S = 
\begin{bmatrix}
    I & 0 \\
    A & B
\end{bmatrix}
$$
where the identity block in the top left corresponds to all the absorbing states. Then it can be shown that 
$$
\lim_{n \to \infty} S^n = 
\begin{bmatrix}
    I & 0 \\
    \sum_{k=0}^\infty B^k A & \lim_{n \to \infty} B^n
\end{bmatrix}
$$
is the limiting distribution.

We show using induction that 
$$
S^n = \begin{bmatrix}
    I & 0 \\
    \sum_{k=0}^{n-1} B^kA & B^n
\end{bmatrix} 
$$
The $n=1$ base case is $S$. For the inductive step, assume that this holds for $i\geq 2$. We show that this is true for $i+1$.

$$
S^{k+1} = S^kS =
\begin{bmatrix}
    I & 0 \\
    \sum_{k=0}^{i-1} B^kA & B^i
\end{bmatrix}
\begin{bmatrix}
    I & 0 \\
    A & B
\end{bmatrix} = 
\begin{bmatrix}
    I & 0 \\
    B^iA + \sum_{k=0}^{i-1} B^kA & B^{i+1}
\end{bmatrix}
$$
as desired. Since the process is time-homogeneous, it is also the stationary distribution. Now we show $\lim_{n\to \infty} B^n = 0$. Consider all the non-absorbing states in the Markov chain: we know that it is possible for them to reach $0$ in a finite number of steps $-ck$, and also that it is possible for them to reach $\beta$ in a finite number of steps. For each state $s$, let $n_s$ be the minimum number of required steps for them to reach an absorbing state. Then let $k = \max_s n_s$. Then in $k$ steps, the probability of reaching an absorbing state is $>0$ for every starting state. Therefore there exists this $k$ such that the maximum row sum of $B^k < 1$, since $S^k$ is row stochastic. 

By Gersgorin's disc theorem that the spectral radius $\rho(S^k) < 1$, since the diagonal entries of $S^k = 0$ (when in a transitory state at time $t$, it is impossible to stay at that state at $t+1$ by design.) We then can write $S^k = P^{-1}JP$, where $J$ is the Jordan normal form. Then 
\begin{equation}\label{limitingmatrix}
\lim_{n \to \infty} S^n = \lim_{n \to \infty} (S^k)^n = \lim_{n \to \infty} P^{-1} J^n P = 0
\end{equation}
because all eigenvalues $\lambda < 1$ in J. Therefore the probability of reaching any absorbing state given that we start at $\pi \in [\beta,1] = 1$.
\end{proof}

\section{Empirical}\label{ap: empirical}

In \autoref{section: empirical demo}, we created a prediction of risk for a sample of applicants drawn from the HMDA 2021 applicant dataset, subsetted to include applicants who were applying for a home loan in Massachusetts. We use the 2021 data, because HMDA datasets before 2020 did not publicly release crucial features, such as the debt-to-income ratio and loan-to-value ratio of applicants. As of 2022, HMDA datasets still do not include another key feature: credit score. A subset of the available features from HMDA can be found in \autoref{tab:feature description}. As described above, HMDA does not contain data about default or late payment. To supplement this dataset, we make use of the FNMA dataset, which does contain information on default and late payment, but does not contain information on crucial demographic characteristics of applicants. Relevant FNMA variables are also listed in \autoref{tab:feature description}.

\begin{table}[]
    \centering
    \begin{tabular}{c|c}
       HMDA 2021 (Massachusetts)  &  FNMA 2017 (Massachusetts) \\
       \hline
        debt-to-income ratio & debt-to-income ratio \\
        loan-to-value ratio & loan-to-value ratio \\
        loan amount & original unpaid balance \\
        & borrower credit score \\
        & loan term \\
        loan purpose & loan purpose \\
        & property type \\
        number of units & number of units \\
        & zip code (first 3 digits) \\
        metropolitan statistical area & metropolitan statistical area \\
        & loan delinquency status \\
        & original listed price \\
        income & \\ 
        applicant gender & \\
        applicant race & \\
        applicant ethnicity & \\
        applicant age & \\
        county code & \\
        action taken (approval or denial) & \\
        interest rate & \\

    \end{tabular}
    \caption{Partial list of data fields in the FNMA and HMDA datasets}
    \label{tab:feature description}
\end{table}

Our methodology involves first constructing a prediction model for risk of late payment, trained on FNMA data, and then using this to predict risk for the HMDA applicants. Because of this, the prediction model can only include features that are shared between the two datasets. \autoref{tab:feature description} shows that six key variables are shared: debt-to-income ratio, loan-to-value ratio, loan amount, number of units, metropolitan statistical area (MSA), and loan purpose. In our methodology, we filtered both datasets to include entries where the loan purpose was "purchase" rather than "refinance." We also did not use MSA in prediction. The remaining four variables were used to construct the predictive model. 

First, using the FNMA dataset, I ran both an OLS and a logistic regression with probability of late payment as the response variable. The results can be seen in \autoref{table:FNMA model}. Comparing regressions (1) and (2) shows that there was a substantial increase in Adjusted $R^2$ when including credit score, which could potentially indicate that credit score contains important explanatory information about financial risk that is not captured by other variables, such as loan-to-value ratio and debt-to-income ratio. This demonstrates a limitation in our study. The ultimate goal of this exercise was to construct an individual loan-level dataset that contained both risk of late payment and demographic information about an applicant's membership in historically marginalized identity categories. Because the former was not available in HMDA, and the latter was not available in FNMA, our prediction for risk of late payment is an imperfect proxy for risk level. For the purposes of this study, however, this proxy is still informative for demonstrating how differences in initial qualification status can be exacerbated in the long term. 

In regression (3), I used the logistic regression (3) from \autoref{table:FNMA model} to predict late payment risk using variables shared between the HMDA and FNMA dataset: loan amount, loan-to-value ratio, debt-to-income ratio, and number of units. This was the final model that was used to generate the predictions for the HMDA dataset, and the results can be seen in \autoref{fig:hmda demo}.

\begin{table}[!htbp] \centering 
  \caption{} 
  \label{table:FNMA model} 
\begin{tabular}{@{\extracolsep{5pt}}lccc} 
\\[-1.8ex]\hline 
\hline \\[-1.8ex] 
 & \multicolumn{3}{c}{\textit{Dependent variable:}} \\ 
\cline{2-4} 
\\[-1.8ex] & \multicolumn{3}{c}{late} \\ 
\\[-1.8ex] & \multicolumn{2}{c}{\textit{OLS}} & \textit{logistic} \\ 
\\[-1.8ex] & (1) & (2) & (3)\\ 
\hline \\[-1.8ex] 
 original\_unpaid\_balance & $-$0.00000$^{***}$ & $-$0.00000$^{***}$ & $-$0.00000$^{***}$ \\ 
  & (0.00000) & (0.00000) & (0.00000) \\ 
  & & & \\ 
 original\_loan\_to\_value\_ratio\_ltv & 0.00001 & 0.001$^{***}$ & 0.010$^{***}$ \\ 
  & (0.0002) & (0.0002) & (0.002) \\ 
  & & & \\ 
 debt\_to\_income\_code & 0.006$^{***}$ & 0.009$^{***}$ & 0.074$^{***}$ \\ 
  & (0.001) & (0.001) & (0.008) \\ 
  & & & \\ 
 borrower\_credit\_score\_at\_origination & $-$0.002$^{***}$ &  &  \\ 
  & (0.0001) &  &  \\ 
  & & & \\ 
 number\_of\_units & 0.025$^{**}$ & 0.031$^{***}$ & 0.244$^{***}$ \\ 
  & (0.010) & (0.010) & (0.079) \\ 
  & & & \\ 
 Constant & 1.427$^{***}$ & 0.029 & $-$2.876$^{***}$ \\ 
  & (0.068) & (0.023) & (0.226) \\ 
  & & & \\ 
\hline \\[-1.8ex] 
Observations & 8,867 & 8,867 & 8,867 \\ 
R$^{2}$ & 0.069 & 0.019 &  \\ 
Adjusted R$^{2}$ & 0.069 & 0.019 &  \\ 
Log Likelihood &  &  & $-$3,327.246 \\ 
Akaike Inf. Crit. &  &  & 6,664.491 \\ 
Residual Std. Error & 0.324 (df = 8861) & 0.332 (df = 8862) &  \\ 
F Statistic & 131.738$^{***}$ (df = 5; 8861) & 43.289$^{***}$ (df = 4; 8862) &  \\ 
\hline 
\hline \\[-1.8ex] 
\textit{Note:}  & \multicolumn{3}{r}{$^{*}$p$<$0.1; $^{**}$p$<$0.05; $^{***}$p$<$0.01} \\ 
\end{tabular} 
\end{table} 

The results are reproduced here in \autoref{fig:hmdakde}, alongside the empirical cdf and maximum mean at different $c$ values in \autoref{fig:emp cdf}. The empirical estimated cdf shows that our assumption of stochastic dominance is satisfied since there is no overlap in the cdf. The maximum means plot shows the maximum attainable average result for each subgroup. For instance, if $c=2$, this simulation shows roughly that the optimal policy will generate $E[\pi^W_t] = 0.94$, and $E[\pi^B_t] = 0.90$. This simulation was based on 25 runs, so $t=25$ here. We can see from this that since our applicant pool has a fairly high initial probability of timely payments, a small $c$ will likely lead to convergence of the two populations to $1$.

\begin{figure}
    \centering
    \includegraphics[scale=0.45]{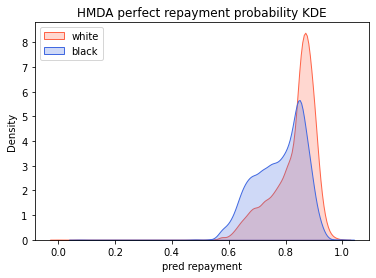}
    \caption{Empirical kernel density estimate of the distribution predicted risk of late payment in the HMDA Massachusetts 2021 data.}
    \label{fig:hmdakde}
\end{figure}

\begin{figure}
    \centering
    \includegraphics[scale=0.45]{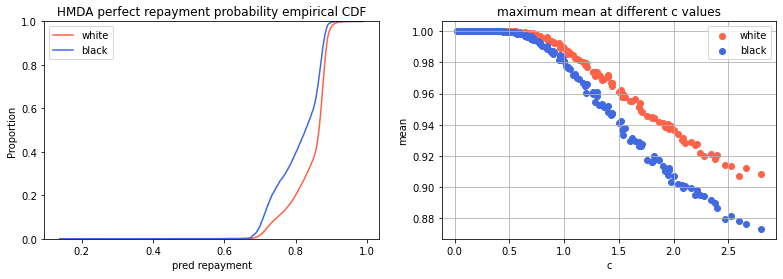}
    \caption{The left panel shows the estimated empirical cdf of the two distributions in \autoref{fig:hmdakde}. The right panel shows the maximum attainable mean for the two populations, given the starting level of $c$.}
    \label{fig:emp cdf}
\end{figure}

\end{document}